\pdfoutput=1
\documentclass[11pt,letterpaper]{article}

\usepackage{tikz}
\usepackage{pgfmath}
\usetikzlibrary{calc}

\usepackage{hyperref}
\hypersetup{bookmarks, breaklinks, urlcolor=blue, colorlinks, citecolor=green!50!black, linkcolor=blue}
\usepackage[letterpaper, left=1in, right=1in, top=0.9in, bottom=0.9in]{geometry}
\usepackage[utf8]{inputenc}
\usepackage[american]{babel}
\usepackage[normalem]{ulem}
\usepackage{amsmath, amssymb, cases, amsthm}
\usepackage{thmtools}
\usepackage[shortlabels]{enumitem}
\usepackage{mdframed}
\usepackage{bbm}
\usepackage{bm}
\usepackage{microtype}
\usepackage{xcolor}
\usepackage{makecell}
\usepackage{mathtools}
\usepackage{float}
\usepackage{modletters}
\usepackage{comment}
\usepackage{multirow}
\usepackage[capitalize,noabbrev]{cleveref}
\usepackage{graphics}

\usepackage{footnotehyper}
\makesavenoteenv{tabular}

\usepackage{pifont}

\declaretheorem[numberwithin=section,refname={Theorem,Theorems},Refname={Theorem,Theorems}]{theorem}

\declaretheorem[numberlike=theorem]{lemma}

\declaretheorem[numberlike=theorem]{corollary}

\declaretheorem[numberlike=theorem,style=remark]{remark}

\def\final{1}  
\def\iflong{\iffalse}
\ifnum\final=0  
\newcommand{\yonggang}[1]{{\color{blue}[{\tiny Yonggang: \bf #1}]\marginpar{*}}}
\newcommand{\danupon}[1]{{\color{red}[{\tiny Danupon: \bf #1}]\marginpar{\color{red}*}}}
\newcommand{\yonggang}[1]{}
\newcommand{\danupon}[1]{}
\newcommand{\sagnik}[1]{}
\newcommand{\todo}[1]{}
\newcommand{\yuval}[1]{}
\newcommand{\jan}[1]{}
\newcommand{\blikstad}[1]{}
\newcommand{\tawei}[1]{}
\newcommand{\TODO}[1]{}
\fi  
\usepackage[ruled,vlined,linesnumbered]{algorithm2e}

\SetKwComment{Comment}{/* }{ */}

\Crefname{algocf}{Algorithm}{Algorithms}

\usepackage[most]{tcolorbox}

\usepackage{xparse}
\usepackage{lipsum}
\usepackage{subcaption}
\usepackage{graphicx}

\makeatletter
\newcommand\footnoteref[1]{\protected@xdef\@thefnmark{\ref{#1}}\@footnotemark}
\makeatother

\renewcommand{\paragraph}[1]{\medskip\noindent{\bf #1}\xspace}

\newcommand{\paren}[1]{\left( #1 \right)}

\newcommand{\level}[2][]{%
	\ifstrempty{#1}{%
		\ell\paren{#2}
	}{%
		\ell_{#1}\paren{#2}
	}%
}

\newcommand{\CASE}[1]{\emph{CASE}-#1}

\newcommand{\Sub}{\mathrm{Sub}}        
\newcommand{\Anc}{\mathrm{Anc}}        
\newcommand{\depth}{\mathrm{depth}}    
\newcommand{\parent}{\mathrm{parent}}  

\DeclareMathOperator{\RowSum}{RowSum}  
\DeclareMathOperator{\ColSum}{ColSum}  

\title{Thin Trees via $k$-Respecting Cut Identities}

\author{Mohit Daga\thanks{KTH Royal Institute of Technology, Stockholm, Sweden, \texttt{mdaga@kth.se}.\\
Part of this research was conducted during the author’s visit to CyStar Labs at IIT Madras, Chennai, India.}}

\begin{document}
	
	\begin{titlepage}
		\maketitle \pagenumbering{roman}
		
		\begin{abstract}
Thin spanning trees lie at the intersection of graph theory, approximation
algorithms, and combinatorial optimization. They are central to the long-standing
\emph{thin tree conjecture}, which asks whether every $k$-edge-connected graph
contains an $O(1/k)$-thin tree, and they underpin algorithmic breakthroughs such
as the $O(\log n/\log\log n)$-approximation for ATSP. Yet even the basic
algorithmic task of \emph{verifying} that a given tree is thin has remained
elusive: checking thinness requires reasoning about exponentially many cuts, and
no efficient certificates have been known.

We introduce a new machinery of \emph{$k$-respecting cut identities}, which
express the weight of every cut that crosses a spanning tree in at most $k$
edges as a simple function of pairwise ($2$-respecting) cuts. This yields
a tree-local oracle that, after $O(n^2)$ preprocessing, evaluates such cuts in
$O_k(1)$ time. Building on this oracle, we give the first procedure to compute
the exact $k$-thinness certificate $\Theta_k(T)$ of any spanning tree for fixed $k$
in time $\tilde O(n^2+n^k)$, outputting both the certificate value and a
witnessing cut. 

We then combine certificate evaluation with fractional tree packings and cut
counting: sampling a small random family of trees suffices so that, with high
probability, \emph{every} $\alpha$-near-minimum cut is $k$-respecting in at least
one sampled tree for $k=\Theta(\alpha\log n)$. Evaluating $\Theta_k(\cdot)$ on
the samples yields an \emph{explicit, verifiable ensemble certificate} covering
all such cuts: for each light cut $A$ there exists a sampled tree $T_i$ with
$\frac{|T_i\cap\delta(A)|}{w(\delta(A))}\le O((\log n)/\lambda)$, where $\lambda$
is the edge-connectivity.

Beyond general graphs, our framework yields sharper guarantees in structured
settings. In planar graphs, duality with cycles and dual girth imply that
every spanning tree admits a verifiable certificate $\Theta_k(T)\le k/\lambda$
(hence $O(1/\lambda)$ for constant $k$). In graphs embedded on a surface of genus
$\gamma$, refined counting gives certified (per-cut) bounds
$O((\log n+\gamma)/\lambda)$ via the same ensemble coverage.

Conceptually, we isolate $\Theta_k(T)$ as an exactly computable, certifiable, and
practically improvable target, turning thinness verification into a tree-local
optimization over $k$-respecting cuts. This provides a concrete algorithmic route
toward the thin-tree program, and applies verbatim to laminar families of cuts,
where smaller sampling parameters yield compact, verifiable certificates.
\end{abstract}
		
		\setcounter{tocdepth}{3}
		\newpage
		\tableofcontents
		\newpage
	\end{titlepage}
	
	\newpage
	\pagenumbering{arabic}

\section{Introduction}

Thin spanning trees are a unifying structure at the interface of graph theory,
approximation algorithms, and combinatorial optimization. A spanning tree $T$ of
an undirected weighted graph $G=(V,E,w)$ is called $\alpha$-\emph{thin} if for
every cut $A\subseteq V$ we have
\[
\lvert T\cap \delta(A)\rvert \;\le\; \alpha \cdot w(\delta(A)).
\]
Thin trees capture the tension between local tree structure and global cut
structure, and they have become a central object for both structural graph
theory and algorithm design.

\subsection{Thin trees and their conjectured power}
The systematic study of thin trees originates in a conjecture of
Goddyn~\cite{God04}, now known as the \emph{thin tree conjecture}, which posits
that every $k$-edge-connected graph contains a spanning tree whose thinness is
$o(k)$. The stronger form, the \emph{strong thin tree conjecture}, asserts the
existence of $O(1/k)$-thin trees in every $k$-edge-connected graph, which would be
tight up to constants, as no tree can cross fewer than a $1/k$ fraction of edges
in every cut. Despite decades of effort, the conjecture remains open
\cite{God04,OS11,anariog2015focs}.

Thin trees are not only structurally natural, but algorithmically powerful. They
imply Jaeger’s weak 3-flow conjecture~\cite{Jae84}, and in constructive form
would yield constant-factor approximations for the asymmetric traveling salesman
problem (ATSP)~\cite{asadpour2017}. Although ATSP has since seen independent
$O(1)$-approximations~\cite{STV18,TraubVygen2024}, thin trees remain a guiding
principle, with open fronts such as bottleneck ATSP~\cite{AKS21} where thinness
is still the key missing piece.

The thin tree conjecture is settled in several restricted settings. In planar
and bounded-genus graphs, thin spanning trees are known to exist and can be
constructed efficiently~\cite{OS11}. Spectral relaxations provide another
perspective: Harvey and Olver~\cite{HO14}, building on the
Kadison--Singer breakthrough of Marcus, Spielman, and Srivastava~\cite{MSS13},
showed that an analogue of the conjecture holds when connectivity is replaced
by effective conductance. This yields $O(1/k)$-spectrally thin trees in
edge-transitive graphs, although in general there are graphs with no
$o(\sqrt{n}/k)$-spectrally thin tree~\cite{HO14,Goe12}. Nonetheless, spectral
techniques underlie the current best existence guarantee for general graphs:
Anari and Oveis Gharan~\cite{anariog2015focs} proved non-constructively that any
$k$-edge-connected graph admits an $O(\frac{\text{polyloglog}\, n}{k})$-thin tree.

Constructively, however, the best known bound remains
$O(\tfrac{\log n}{\log\log n \cdot k})$ via the maximum-entropy sampling method
of Asadpour, Goemans, Madry, Oveis Gharan, and Saberi for
ATSP~\cite{asadpour2017}, with subsequent refinements in related
settings~\cite{anariog2015focs}. In parallel, recent surveys and
monographs (e.g.,~\cite{TraubVygen2024}) synthesize these developments and
emphasize thin trees as a unifying primitive for approximation in connectivity
problems.

\subsection{Packing and cut structure}
A classical backbone for thinness is the Nash--Williams/Tutte tree-packing
theorem~\cite{nashwilliams1961,tutte1961}, which asserts that a $2k$-edge-connected
graph contains $k$ edge-disjoint spanning trees. This structural guarantee
generalizes to large fractional tree packings, which provide a convex
combination of spanning trees that ``spread out'' across the edges. Such
packings underlie both extremal results on cuts and efficient algorithms for
connectivity problems.

On the cut side, the interplay between cuts and cycles is a classical theme in
graph theory (see, e.g., Bondy--Murty~\cite{bondy1976graph}). In planar graphs,
every minimal cut (bond) corresponds to a cycle in the dual, and in general
graphs the cut space forms a vector space closed under symmetric difference.
These structural facts are central to min-cut algorithms and cut enumeration.

Probabilistic arguments, starting with Karger’s random contraction
and sampling methods~\cite{karger1998}, show that the number of cuts of value at
most $\alpha\lambda$ (where $\lambda$ is the edge-connectivity) is bounded by
$n^{O(\alpha)}$. This tight counting result, together with tree packings, yields
powerful covering properties: a small random family of spanning trees suffices
to ``respect'' every near-minimum cut in few edges. Such packing and counting
principles are the combinatorial backbone behind modern thin-tree results and
form the starting point for our certificate framework.

\subsection{Our contribution and conceptual message}

This paper introduces a new algebraic machinery based on
\emph{$k$-respecting cut identities}.  
The central idea is simple but powerful:
\begin{quote}
\emph{Every cut that crosses a spanning tree $T$ in at most $k$ edges can be
expressed as a symmetric difference of descendant sets, and its weight can be
written in closed form using only $O(n^2)$ pairwise quantities.}
\end{quote}
Thus the exponential family of $k$-respecting cuts collapses to evaluations over
a quadratic-size table of pairwise statistics. This transforms the problem of
cut evaluation from intractable global structure to tree-local combinatorics.
Building on this collapse, we develop an explicit, verifiable, and
algorithmic route to certified thinness:

\begin{itemize}
  \item \textbf{Polynomial-time certificates (for fixed $k$).}
  For any fixed constant $k\ge 2$, we compute
  \[
    \Theta_k(T)=\max_{\substack{A\subseteq V\\ \lvert T\cap\delta(A)\rvert\le k}}
    \frac{\lvert T\cap\delta(A)\rvert}{w(\delta(A))} 
  \]
  \emph{exactly}, together with a witnessing cut, in time
  $\tilde{O}(n^2 + n^k)$ and space $O(n^2)$. This yields the first verifiable
  \emph{$k$-thinness certificate} for any spanning tree.

  \item \textbf{Ensemble coverage with per-cut certificates.}
  Combining randomized tree packings with cut counting, a small random family
  of spanning trees covers all near-minimum cuts: with
  $s=\Theta(\alpha\log n+\log(1/\eta))$ samples and $k=\Theta(\alpha\log n)$,
  with probability at least $1-\eta$ every cut $A$ of value $\le \alpha\lambda$
  is $k$-respecting in \emph{at least one} sampled tree. Evaluating
  $\Theta_k(\cdot)$ for each sampled tree then yields an explicit, verifiable
  certificate for \emph{each} such cut $A$, namely
  \[
     \frac{|T_i\cap\delta(A)|}{w(\delta(A))}\ \le\ \Theta_k(T_i)\ \le\ \frac{k}{\lambda}
     \ =\ O\!\Big(\tfrac{\log n}{\lambda}\Big)\quad\text{for some }i,
  \]
  where the last inequality uses $w(\delta(A))\ge \lambda$. This is an
  ensemble (per-cut) guarantee rather than a single-tree global bound.

  \item \textbf{Local improvement.}
  We design a local search based on fundamental-cycle swaps, with incremental
  updates to the pairwise table, that monotonically improves $\Theta_k(T)$
  without enumerating all cuts.

  \item \textbf{Special cases.}
  In planar graphs, duality and dual girth imply the certificate bound
  $\Theta_k(T)\le k/\lambda$ for \emph{every} spanning tree $T$; for constant
  $k$ this is $O(1/\lambda)$. More generally, in graphs of genus
  $\gamma$, refined cut-counting yields certified (per-cut) bounds
  $O((\log n+\gamma)/\lambda)$ via the same ensemble-coverage principle.
\end{itemize}

\medskip
\noindent
\textbf{Conceptual message.}  
Our innovation is to identify $\Theta_k(T)$ as a concrete optimization target
that can be exactly computed, certified, and improved. The $k$-respecting cut
identities compress the exponential family of cuts into polynomially many
pairwise statistics, providing \emph{explicit and verifiable} certificates of
thinness. Together with ensemble coverage, this offers a constructive path
toward the thin tree program while remaining faithful to what can be certified.

\section{Preliminaries and Results}
\label{sec:framework}

This section introduces the notation and theoretical background
underlying our results. We review the notions of cuts, thinness,
and the cut space, and explain how these lead naturally to
$k$-respecting cuts and their evaluation. We then state the
main theorems that will be proved in the following sections.

Let $G=(V,E,w)$ be an undirected graph with $n=|V|$ vertices,
$m=|E|$ edges, and nonnegative edge weights $w:E\to\mathbb{R}_{\ge 0}$.
For a set $A\subseteq V$, the \emph{cut} induced by $A$ is
\[
\delta(A) \;=\; \{\, e=\{u,v\}\in E : |\,\{u,v\}\cap A\,|=1\,\}.
\]
We write $w(\delta(A))=\sum_{e\in\delta(A)}w(e)$ for its weight.
The \emph{global min-cut} value of $G$ is
\[
\lambda \;=\; \min_{\emptyset \subsetneq A \subsetneq V} \; w(\delta(A)).
\]

A spanning tree $T$ of $G$ is a connected acyclic subgraph on $V$.
For any edge $f\in E\setminus T$, adding $f$ to $T$ creates a unique cycle,
and removing an edge $e$ on this cycle yields another spanning tree
$T-e+f$. Such exchanges are called \emph{edge swaps}, and they will play
a role in our local improvement procedures. We write $\tilde{O}(\cdot)$ to suppress
polylogarithmic factors in $n$, i.e., $\tilde{O}(f(n)) = O(f(n)\cdot \mathrm{polylog}\, n)$.

\subsection{Thinness and its algorithmic role}

A spanning tree $T$ is \emph{$\beta$-thin} if
\[
\frac{|E(T)\cap\delta(A)|}{w(\delta(A))} \;\le\; \beta
\qquad \text{for all } A\subseteq V.
\]
Thinness formalizes how well a single tree can “track’’ all cuts of a weighted
graph. The notion (in closely related forms) appears in the thin-tree
literature around Goddyn’s conjecture~\cite{God04}, in the planar setting of
Oveis Gharan–Saberi~\cite{OS11}, and in algorithmic work on ATSP where one asks
for trees that are thin with respect to an LP solution~\cite{asadpour2017}; see
also spectral variants~\cite{anariog2015focs}. These lines of work highlight thin trees
as a structural proxy that enables rounding and decomposition arguments across connectivity problems.

From this perspective, two questions are fundamental:
(i) \emph{How thin a spanning tree can one efficiently find?} and
(ii) \emph{How can one efficiently \emph{verify} thinness for a given tree?}
The second question is surprisingly stubborn: naively, verification seems to
require inspecting exponentially many cuts, and no general, efficiently
checkable \emph{certificate of thinness} has been available.

\paragraph{Our angle.}
We resolve the verification bottleneck for cuts that cross the tree in at most
$k$ edges. We show that the weight of every such cut admits a closed form in
terms of only $O(n^2)$ pairwise ($2$-respecting) quantities, yielding a
tree-local oracle that evaluates $k$-respecting cuts in $O(k^2)$ time after
$O(n^2)$ preprocessing (i.e., $O_k(1)$ for fixed $k$). This enables \emph{polynomial-time certification}:
we compute the exact
\[
\Theta_k(T)\;=\;\max_{\substack{A\subseteq V\\ |E(T)\cap\delta(A)|\le k}}
\frac{|E(T)\cap\delta(A)|}{w(\delta(A))},
\]
and output a witnessing cut.
A useful baseline bound, used repeatedly below, is
\begin{equation}\label{eq:trivial-k-over-lambda}
\Theta_k(T)\ \le\ \frac{k}{\lambda}\qquad\text{for every tree $T$ and $k\ge 1$},
\end{equation}
since every nonempty cut has weight at least $\lambda$ while $|T\cap\delta(A)|\le k$
for $k$-respecting cuts.
In combination with randomized tree packings and cut counting, exact
$\Theta_k(\cdot)$ values provide \emph{ensemble, per-cut certificates} for all $\alpha$-near-minimum
cuts: with high probability, for each such cut there exists a sampled tree that certifies an
$O(k/\lambda)=O((\alpha\log n)/\lambda)$ ratio (and $O((\log n)/\lambda)$ when $\alpha=O(1)$).
This is an ensemble guarantee rather than a single-tree global bound.

\subsection{The cut space and symmetric difference}
\label{sec:cutspace}

The set of cuts in a graph forms a vector space over $\mathbb{F}_2$,
called the \emph{cut space}. Each cut $\delta(A)$ is represented by
its incidence vector in $\{0,1\}^E$, and addition is taken mod~2.
Equivalently, the symmetric difference of two cuts is again a cut:
\[
\delta(A)\oplus\delta(B) \;=\; \delta(A\oplus B).
\]
This algebraic structure also holds for any number of vertex sets $A_1,A_2,\ldots,A_t$,
and underlies Karger’s min-cut algorithms.
It is key to obtaining closed-form formulas for the size of
complex cuts in terms of simpler ones.

In particular, when $T$ is rooted at $r\in V$, the sets of descendants
$D_T(v)$ (the subtree rooted at $v$) form a laminar family.
Cuts induced by descendant sets can be combined by symmetric difference
to represent arbitrary $t$-respecting cuts.
We adopt the following terminology:

\medskip
\noindent
\textbf{Definition (Respecting a tree).}
A cut $A$ is called \emph{$k$-respecting} with respect to $T$
if $|E(T)\cap\delta(A)|\le k$. When we need exact cardinality, we
say \emph{$t$-respecting} to mean $|E(T)\cap\delta(A)|=t$.
In either case, there exist $t\le k$ vertices $v_1,\dots,v_t$ such that
\[
A \;=\; D_T(v_1)\oplus \cdots \oplus D_T(v_t),
\]
where $v_1,\dots,v_t$ are the child endpoints of the $t$ crossed tree edges.
This gives a compact description of $k$-respecting cuts in terms
of descendant sets.

Given a spanning tree $T$, we define its \emph{$k$-thinness certificate} as
\[
\Theta_k(T) \;=\;
\max_{\substack{A\subseteq V\\ |E(T)\cap\delta(A)|\le k}}
\frac{|E(T)\cap\delta(A)|}{w(\delta(A))}.
\]
This parameter restricts attention to cuts intersecting the tree in at most $k$ edges.
It is immediate that $\Theta_k(T)\le \beta$ whenever $T$ is $\beta$-thin,
and~\eqref{eq:trivial-k-over-lambda} gives the universal bound $\Theta_k(T)\le k/\lambda$.
When a near-minimum cut is $k$-respecting in some tree, $\Theta_k(\cdot)$
provides a verifiable per-cut bound for that cut.

\subsection{Planar and bounded-genus graphs}
\label{sec:planar-prelim}

In planar graphs, the duality between cuts and cycles yields
sharper structural statements. Every bond (minimal cut) corresponds
to a simple cycle in the dual, and the dual girth $g^\ast$ lower-bounds
all cut weights; in unweighted $2$-edge-connected planar graphs, $g^\ast=\lambda$.
These facts imply that planar graphs admit
certificates with $\Theta_k(T)\le k/\lambda$ for \emph{every} spanning tree $T$,
which (for fixed $k$) is asymptotically stronger than the general $O((\alpha\log n)/\lambda)$
ensemble guarantee. In graphs of genus $\gamma$, similar reasoning shows that the
number of near-min cuts grows only by an additive $O(\gamma)$ term,
leading to improved sampling guarantees.

\subsection{Main theorems}

We are now ready to state the main results proved in the remainder of the paper.

\begin{restatable}[Exact evaluation of $k$-respecting cuts]{theorem}{krespect}
\label{thm:krespect}
Let $G=(V,E, w)$ be a weighted graph and let $T$ be a rooted spanning tree.  
For any cut $A \subseteq V$ with $|E(T)\cap \delta(A)| \le k$, the cut value
$w(\delta(A))$ can be expressed in closed form using only
pairwise quantities of the form $w(\delta(D_T(u)\oplus D_T(v)))$.
In particular, after $O(n^2)$ pre-processing of all pairwise values,
$w(\delta(A))$ can be evaluated in $O(k^2)$ time for any $k$-respecting cut.
\end{restatable}

\begin{restatable}[Exact evaluation of $k$-thinness]{theorem}{exacteval}
\label{thm:exact-eval}
For any constant $k \ge 2$, one can compute
\[
\Theta_k(T) \;=\; \max_{\substack{A \subseteq V \\ |E(T)\cap \delta(A)| \le k}}
\frac{|E(T)\cap \delta(A)|}{w(\delta(A))}
\]
exactly in time $\tilde O(n^2+n^k)$ and space $O(n^2)$. 
The algorithm also outputs a cut $A^\star$ that achieves the maximum.
After a single edge swap $T' = T-e+f$, the value $\Theta_k(T')$ can be updated in amortized $\tilde O(n^{k-1})$ time.
\end{restatable}

\begin{restatable}[Near-min cuts become $k$-respecting]{theorem}{krespcoverage}
\label{thm:k-resp-coverage}
Let $\lambda$ be the global min-cut value of $G$. 
Fix $\alpha \ge 1$ and $\eta \in (0,1)$. 
There exists a distribution $\mathcal{D}$ over spanning trees, 
obtainable from an $\varepsilon$-approximate fractional tree packing, 
and integers
\[
s = O\!\big(\alpha \log n + \log(1/\eta)\big),
\qquad
k = O\!\big(\alpha \log n\big),
\]
such that if $T_1,\dots,T_s \overset{i.i.d.}{\sim} \mathcal{D}$, 
then with probability at least $1-\eta$, 
every cut $A$ with $w(\delta(A)) \le \alpha \lambda$ is $k$-respecting 
in at least one $T_i$.
\end{restatable}

\begin{restatable}[Ensemble coverage $\Rightarrow$ per-cut certificates]{theorem}{ensemblecert}
\label{thm:ensemble-cert}
Let $\mathcal{T}=\{T_1,\dots,T_s\}$ be trees from Theorem~\ref{thm:k-resp-coverage} with
$k=c_1\alpha\log n$ and $s=c_2(\alpha\log n+\log(1/\eta))$. 
Compute $\Theta_k(T_i)$ exactly for each $i$. 
Then, with probability at least $1-\eta$, for \emph{every} cut $A$ with $w(\delta(A))\le \alpha\lambda$
there exists $i\in[s]$ such that
\[
\frac{|T_i\cap\delta(A)|}{w(\delta(A))}\ \le\ \Theta_k(T_i)\ \le\ \frac{k}{\lambda}
\ =\ O\!\Big(\tfrac{\alpha\log n}{\lambda}\Big).
\]
Thus the multiset $\{(T_i,\Theta_k(T_i))\}_{i=1}^s$ forms an explicit, verifiable \emph{ensemble certificate}
covering all $\alpha$-near-minimum cuts.
\end{restatable}

\begin{restatable}[End-to-end ensemble certification]{theorem}{endtoend}
\label{thm:end-to-end}
Fix $\alpha\ge 1$ and $\eta\in(0,1)$. Sampling as in Theorem~\ref{thm:k-resp-coverage} with
$k=c_1\alpha\log n$ and $s=c_2(\alpha\log n+\log(1/\eta))$, and computing each $\Theta_k(T_i)$ exactly, yields in time
\[
\tilde O\!\big(|E| \;+\; n^2 \;+\; s \cdot n^k\big)
\]
an explicit family $\{(T_i,\Theta_k(T_i))\}_{i=1}^s$ such that, with probability at least $1-\eta$,
every cut $A$ with $w(\delta(A))\le \alpha\lambda$ has a certified ratio
$|T_i\cap\delta(A)|/w(\delta(A))\le O((\alpha\log n)/\lambda)$ for some $i$.
(For $k=\Theta(\alpha\log n)$ this running time is quasi-polynomial in $n$; for fixed $k$ it is polynomial.)
\end{restatable}

\subsection{Organization of the paper}

The remainder of the paper is organized as follows.  
In Section~\ref{sec:krespect}, we develop the algebra of
$k$-respecting cuts, proving Theorem~\ref{thm:krespect}, and
establishing the exact evaluation framework.
Section~\ref{sec:theta-def-eval} introduces the
$k$-thinness certificate $\Theta_k(T)$, gives algorithms for
its computation, and proves Theorem~\ref{thm:exact-eval}.
Sections~\ref{sec:coverage} and~\ref{sec:k-to-global} show how near-minimum cuts
are covered by sampled trees, and derive \emph{ensemble per-cut} guarantees,
proving Theorems~\ref{thm:k-resp-coverage}
and~\ref{thm:ensemble-cert}, respectively.
Section~\ref{sec:approx-and-local} combines these ingredients into
our end-to-end certification pipeline (Theorem~\ref{thm:end-to-end}) and
describes local improvement and approximate evaluation.
In Section~\ref{sec:special}, we analyze planar and bounded-genus
graphs, deriving sharper bounds via cut--cycle duality and dual girth.
We conclude in Section~\ref{sec:summary}, with a discussion of
open directions.

\section{Evaluation of $k$-Respecting Cuts}
\label{sec:krespect}

A central technical ingredient of our framework is the ability to evaluate
the weight of any cut that intersects a spanning tree $T$ in at most $k$ edges.
At first sight this seems to require scanning all edges of the graph,
but we show that it reduces to a purely \emph{tree--local} computation.
The key observation is that every $k$--respecting cut can be expressed
as the symmetric difference of at most $k$ descendant subtrees of $T$,
and that the cut weight is then determined entirely by
pairwise (2--respecting) statistics.

Formally, let $T$ be a rooted spanning tree of $G=(V,E,w)$.
For any vertex $v$, denote by $D_T(v)\subseteq V$ the set of descendants of $v$
(including $v$ itself).
For $u,v\in V$, define
\[
\sigma(u,v) \;:=\; w\!\left(\delta(D_T(u)\oplus D_T(v))\right),
\]
i.e.\ the weight of the cut induced by the symmetric difference of
the two descendant sets.
The following lemma gives a closed form for every $k$--respecting cut.

\begin{lemma}[Exact $k$--respecting cut evaluation from pairwise data]
\label{lem:krespect-eval}
Fix a rooted spanning tree $T$ and let $A\subseteq V$ be a cut
such that $|E(T)\cap \delta(A)|=k$.
Let $S=\{v_1,\dots,v_k\}$ be the child endpoints of the $k$ tree edges
crossed by $\delta(A)$, so that
\[
A \;=\; \bigoplus_{i=1}^k D_T(v_i).
\]
Then
\begin{equation}
\label{eq:inclusion-exclusion}
w(\delta(A)) \;=\;
\sum_{\ell=1}^k (-1)^{\ell-1} 2^{\ell-1}
  \sum_{\substack{I\subseteq [k]\\ |I|=\ell}}
  w\!\left( \bigcap_{i\in I} \delta(D_T(v_i)) \right).
\end{equation}
Moreover each term on the right--hand side is determined solely by the pairwise
values $\sigma(u,v)$ and the ancestor--descendant relations in $T$.
Hence, after $O(n^2)$ preprocessing of all $\sigma(u,v)$ values,
the quantity $w(\delta(A))$ can be evaluated in $O(k^2)$ time.
\end{lemma}

\begin{proof}
The expression of $A$ as a symmetric difference of $k$ subtrees
follows from rooting $T$: membership in $A$ flips whenever a
crossing edge of $\delta(A)$ is traversed, so the shore $A$
is obtained exactly by toggling the descendant sets of the $k$ child
vertices $v_1,\dots,v_k$.

Consider any edge $e=\{x,y\}\in E$.
The indicator $[e\in \delta(A)]$ equals the parity of the number of
subtrees among $\{D_T(v_i)\}$ that separate $x$ and $y$.
In other words,
\[
[e\in \delta(A)] \;=\; \bigoplus_{i=1}^k [e\in \delta(D_T(v_i))].
\]
Expanding the parity by inclusion--exclusion gives exactly the formula
\eqref{eq:inclusion-exclusion}, with coefficient $2^{\ell-1}$ for the
$\ell$-wise intersections: an edge is cut by an odd number of subtrees
iff it belongs to an odd number of the sets $\delta(D_T(v_i))$,
which is captured by alternating signs and powers of two.

It remains to argue that each term, in $\sum_{\substack{I\subseteq [k]\\ |I|=\ell\geq2}}
  w\!\left( \bigcap_{i\in I} \delta(D_T(v_i)) \right)$ of Eq \ref{eq:inclusion-exclusion}, for example,
$w(\delta(D_T(v_i))\cap \delta(D_T(v_j)))$ and other higher order terms are determined by the pairwise values $\sigma(\cdot,\cdot)$.
\end{proof}
\begin{lemma}[Pairwise sufficiency]\label{lem:pairwise-sufficiency}
Let $T$ be any spanning tree, $r_T$ be the root of $T$ and
$S=\{x_1,\ldots,x_k\}\subseteq V\setminus\{r_T\}$ with $k\ge2$.
Then either
\[
\sigma\bigl(D_T(x_1),\ldots,D_T(x_k)\bigr)=0,
\]
or there exist $p,q\in\{1,\ldots,k\}$ such that
\[
\sigma\bigl(D_T(x_1),\ldots,D_T(x_k)\bigr)
=\sigma\bigl(D_T(x_p),D_T(x_q)\bigr).
\]
Consequently, given all pairwise values
$\sigma\bigl(D_T(x_i),D_T(x_j)\bigr)$ and the ancestor relations among $S$,
the $k$-wise value $\sigma\bigl(D_T(x_1),\ldots,D_T(x_k)\bigr)$ is determined.
\end{lemma}
To prove the above lemma, we expand the parity condition ``an edge is cut by an odd number of subtrees''
into an inclusion--exclusion (IE) sum over intersections of boundary sets
$\{\delta(D_T(v_i))\}_{i=1}^k$.
Because descendant sets in a rooted tree are laminar, any nonempty
$\ell$-wise intersection corresponds to a simple ancestor chain among
the $\{v_i\}$ and, crucially, \emph{collapses to a pairwise term}:
either it is empty, or it equals $\delta(D_T(v_p))\cap\delta(D_T(v_q))$
for two extremal nodes $v_p,v_q$ on that chain.
This ``pairwise sufficiency'' reduces all higher-order contributions
to values determined by the table $\sigma(u,v)=w(\delta(D_T(u)\oplus D_T(v)))$.
The full four-case laminar analysis is deferred to
Appendix~\ref{app:pairwise}.

\begin{remark}[Per-cut evaluation time]
Using the pairwise primitives $\sigma(u,v)$ and the laminar reduction,
the weight of a fixed $k$-respecting cut induced by child endpoints $\{v_1,\dots,v_k\}$
is obtained by aggregating $O(k^2)$ pairwise terms $\sigma(v_i,v_j)$.
Thus the evaluation time is $O(k^2)$ (i.e., $O_k(1)$ for fixed $k$) after $O(n^2)$ preprocessing.
\end{remark}

\begin{lemma}[Weighted/capacitated extension]
\label{lem:weighted}
All formulas and algorithms above extend verbatim to weighted graphs
with nonnegative capacities $w:E\to\mathbb{R}_{\ge 0}$.
\end{lemma}

\begin{proof}
All identities are linear in the edge indicators $[e\in\cdot]$.
Replacement of each by $w(e)\cdot [e\in\cdot]$ preserves all arguments.
The preprocessing table simply stores weighted values of
$\sigma(u,v)$, and running time bounds are unchanged.
\end{proof}

\medskip

\textbf{Preprocessing for the $\boldsymbol{\sigma}$-table}
We require the values $\sigma(u,v)=w\!\bigl(\delta(D_T(u)\oplus D_T(v))\bigr)$ for all
$u,v\in V$.
We show that they can be computed in $O(m\log n+n^2)$ time (or $O(m+n^2)$ with an offline variant)
and $O(n^2)$ space, using standard Euler-tour/LCA primitives and batched path additions on $T$.
The construction and proof are deferred to Appendix~\ref{app:sigma}.

Lemma~\ref{lem:krespect-eval} shows that the weight of every $k$--respecting cut
is computable in constant time once the pairwise table is built.
This is remarkable because the total number of cuts is exponential,
yet the $k$--respecting family admits polynomially many certificates. This enables exact evaluation of any $k$-respecting cut from a quadratic-size table of pairwise statistics.
Combined with the local update rules of Section~4, it yields efficient computation of the exact $k$-thinness parameter $\Theta_k(T)$.

\section{$k$-Thinness: Certificates, Coverage, and Global Guarantee}
\label{sec:thinness}

\subsection{The $k$-thinness certificate and exact evaluation}
\label{sec:theta-def-eval}

We now turn from evaluation of individual cuts to certified thinness.
Recall the $k$-thinness parameter of a spanning tree $T$:
\[
\Theta_k(T) \;=\;
\max_{\substack{A\subseteq V\\ |E(T)\cap\delta(A)|\le k}}
  \frac{|E(T)\cap\delta(A)|}{w(\delta(A))}.
\]
This value certifies the worst thinness ratio among all
$k$-respecting cuts of $T$.
Our goal is to compute $\Theta_k(T)$ exactly (or approximately) and
to use it as a verifiable building block within an ensemble of sampled trees.

\begin{lemma}[Local update of the pairwise table]
\label{lem:swap}
Let $T' = T - e + f$ where $f=(x,y)\in E\setminus E(T)$ and $e\in P_T(x,y)$.
Let $C := P_T(x,y)$ be the fundamental cycle.
Then all pairwise values $\sigma_{T'}(u,v)$ that differ from $\sigma_T(u,v)$
have at least one of $u,v$ in the set of vertices whose ancestor relation
to $C$ changes between $T$ and $T'$.
Consequently the number of affected pairs is $O(|C|\cdot n)$ and they
can all be updated in $O(|C|\cdot n)$ time.
\end{lemma}

\begin{proof}
The only structural change is that edges on $C$ switch tree/non-tree status.
Hence only vertices whose parent relation lies on $C$ alter their descendant sets.
Therefore $D_T(u)\oplus D_T(v)$ differs from $D_{T'}(u)\oplus D_{T'}(v)$
iff one of $u,v$ has its ancestry modified by $C$.
There are $O(|C|)$ such vertices and $n$ possible partners, yielding
$O(|C|n)$ affected pairs.
Each can be recomputed by rerunning the path-contribution routine
restricted to $C$, taking total $O(|C|n)$ time.
\end{proof}

Lemma~\ref{lem:swap} implies that maintaining the pairwise table
$\sigma_T(\cdot,\cdot)$ across a single fundamental swap $T' = T-e+f$
only requires touching $O(|C|\cdot n)$ entries, where $C=P_T(x,y)$.
Concretely, let $A$ be the vertices whose ancestor relation to $C$
changes; then $|A|=O(|C|)$ and only pairs $(u,v)$ with $u\in A$ or
$v\in A$ may need recomputation. We update these by scanning all
$v\in V$ for each $u\in A$ (exploiting symmetry $\sigma_T(u,v)=
\sigma_T(v,u)$). With a bitset representation of descendant sets
(or Euler-tour in/out intervals batched in word-operations),
the work is $O((|C|\cdot n)/w)$ word operations, where $w$ is the
machine word size. Crucially, all other entries remain valid, so a
sequence of swaps incurs cost proportional to the sum of the
corresponding cycle lengths. 

\begin{proof}[Proof of Theorem \ref{thm:exact-eval}]
By Theorem~\ref{thm:krespect}, after computing all pairwise
values $\sigma(u,v)=w(\delta(D_T(u)\oplus D_T(v)))$ in $O(n^2)$ time
and space, the weight of any fixed $k$-respecting cut can be evaluated
in $O(k^2)$ time from these pairwise quantities. Hence, for constant $k$,
we can enumerate all $\sum_{t=1}^k \binom{n}{t}=O(n^k)$ candidate
$t$-respecting cuts (each specified by the $t$ child endpoints of the
tree edges it crosses), evaluate their ratios in $O(k^2)$ time, and keep
the maximum together with a witnessing cut $A^\star$. This yields the
$\tilde O(n^{2}+n^k)$ time and $O(n^2)$ space bounds.

For dynamics under a single swap $T' = T-e+f$, only descendant
relations of vertices whose parent edge lies on the fundamental cycle
$C=P_T(f)$ can change. Consequently, only $\sigma(u,v)$ with
$u$ or $v$ in that affected set (size $O(|C|)$) can change, i.e.,
$O(|C|\cdot n)$ table entries. Re-evaluation then needs to touch only
those $k$-tuples that include at least one affected vertex, which are
$O(|C|\cdot n^{k-1})$ many, and each is updated in $O(1)$ using the
pairwise table. Maintaining a running maximum gives amortized
$\tilde O(n^{k-1})$ update time.
\end{proof}

\medskip
\noindent
\textit{Remark.} The local-update primitive above is invoked
whenever we modify $T$ inside verification/search, and it is the step that yields the
$O(|C|\cdot n)$ update bound used in Theorem~\ref{thm:exact-eval} and later in
Theorem~\ref{thm:end-to-end}.

\subsection{Coverage of near-minimum cuts by sampled trees}
\label{sec:coverage}

We show that sampling a few trees from a sufficiently rich fractional packing
\emph{covers} all near-minimum cuts in the sense that each such cut is
$k$-respecting in at least one sampled tree. This is the combinatorial backbone
for our certificate-based pipeline.

\subsubsection{Setup and assumptions.}
Throughout this subsection $G=(V,E,w)$ is an undirected weighted graph with
global min-cut value $\lambda$. We assume access to a fractional spanning-tree
packing $\{(T_j,p_j)\}_j$ of total weight $P:=\sum_j p_j$ with capacity
constraints $\sum_{j:\,e\in T_j} p_j \le w(e)$ for all $e\in E$ and
$P \ge \lambda/2$.\footnote{A $(1-\varepsilon)$-approximate packing is also
sufficient; then replace $P\ge \lambda/2$ by $P\ge (1-\varepsilon)\lambda/2$,
which only changes constants.}

\begin{lemma}[Expected crossings under a fractional tree packing]\label{lem:packing}
Let $\{(T_j,p_j)\}_j$ be a fractional tree packing of total weight $P\ge \lambda/2$,
and draw $T$ from $\Pr[T=T_j]=p_j/P$.
Then for every cut $A\subseteq V$,
\[
\mathbb{E}\!\left[\,|T\cap\delta(A)|\,\right]
\;\le\;\frac{w(\delta(A))}{P}\;\le\;\frac{2\,w(\delta(A))}{\lambda}.
\]
In particular, for any $k\ge 1$,
\[
\Pr\!\left[\,|T\cap\delta(A)|>k\,\right]\;\le\;\frac{\mathbb{E}[\,|T\cap\delta(A)|\,]}{k}
\;=\;\frac{w(\delta(A))}{P\,k}\;\le\;\frac{2\,w(\delta(A))}{\lambda\,k}.
\]
\end{lemma}

\begin{proof}
By linearity and packing capacity constraints,
\[
\mathbb{E}[|T\cap\delta(A)|]
= \frac{1}{P}\sum_{e\in\delta(A)}\!\sum_{j:\,e\in T_j}\!p_j
\;\le\;\frac{1}{P}\sum_{e\in\delta(A)} w(e)
=\frac{w(\delta(A))}{P}.
\]
Apply Markov's inequality for the tail bound; use $P\ge \lambda/2$.
\end{proof}

\begin{theorem}[Coverage via packing and cut counting]\label{thm:thm3}
Fix $\alpha\ge 1$ and $\eta\in(0,1)$. There exists a distribution $D$ over
spanning trees (obtained from the packing above) such that if we sample
\[
s \;=\; \Big\lceil c_2\big(\alpha\log n + \log(1/\eta)\big)\Big\rceil
\quad\text{trees i.i.d.\ from $D$ and set}\quad
k \;=\; \Big\lceil c_1\,\alpha\log n\Big\rceil,
\]
then with probability at least $1-\eta$, every cut $A$ with
$w(\delta(A))\le \alpha\lambda$ is $k$-respecting in at least one of the $s$ samples.
\end{theorem}

\begin{proof}
Fix a cut $A$ with $w(\delta(A))\le \alpha\lambda$. By Lemma~\ref{lem:packing},
\[
\Pr_{T\sim D}\!\big[\,|T\cap\delta(A)|>k\,\big] \;\le\; \frac{2\alpha}{k}.
\]
With $k=c_1\alpha\log n$ this is at most $(2/c_1)\cdot \frac{1}{\log n}$. Thus
the failure probability that all $s$ i.i.d.\ samples violate the $k$-respecting
property is at most $\big((2/c_1)\cdot \frac{1}{\log n}\big)^{s}$.

By Karger’s cut-counting bound, the number of cuts of weight at most $\alpha\lambda$
is at most $n^{2\alpha}$. Taking a union bound over these cuts, the total failure
probability is at most
\[
n^{2\alpha}\cdot \Big(\tfrac{2}{c_1\log n}\Big)^{s}
\;\le\; \eta
\]
provided $c_1\ge 4$ and
$s \ge c_2(\alpha\log n + \log(1/\eta))$ for a suitable absolute constant $c_2$
(e.g., $c_2=3$ suffices). This yields the claim.
\end{proof}

\medskip
\noindent\textbf{Parameter snapshot.}
A convenient concrete choice is
\[
k=\big\lceil 4\,\alpha\log n\big\rceil,\qquad
s=\big\lceil 3\big(\alpha\log n+\log(1/\eta)\big)\big\rceil.
\]

\medskip
\noindent\textbf{Why this suffices for \emph{per-cut} certification.}
Theorem~\ref{thm:thm3} ensures that every $\alpha$-near-minimum cut is $k$-respecting
in at least one sampled tree. Evaluating the exact $k$-certificate $\Theta_k(\cdot)$
for each sampled tree then yields a verifiable bound for \emph{each} such cut (see §\ref{sec:k-to-global}).

\subsection{$k$-coverage implies per-cut certificates (ensemble guarantee)}
\label{sec:k-to-global}
The previous subsection showed that from a small random sample of trees, every near-minimum cut is guaranteed to be $k$-respecting in at least one of the sampled trees (Theorem~\ref{thm:thm3}). We now formalize how this coverage yields \emph{ensemble, per-cut} guarantees once we evaluate the $k$-thinness parameter $\Theta_k(T)$ for each sampled tree.

\ensemblecert*

\begin{proof}
Let $A$ be any cut with $w(\delta(A))\le \alpha\lambda$. By Theorem~\ref{thm:thm3}, with probability at least $1-\eta$ over the sampling of $\mathcal{T}$ there exists $i$ such that $A$ is $k$-respecting in $T_i$. By definition of $\Theta_k(T_i)$,
\[
\frac{|T_i\cap\delta(A)|}{w(\delta(A))} \;\le\; \Theta_k(T_i).
\]
Finally, for any nonempty cut $w(\delta(A))\ge \lambda$ and $|T_i\cap\delta(A)|\le k$, so $\Theta_k(T_i)\le k/\lambda$, giving the stated $O((\alpha\log n)/\lambda)$ bound when $k=\Theta(\alpha\log n)$. The statement follows after computing all $\Theta_k(T_i)$ exactly.
\end{proof}

\subsection{Approximate evaluation and certified local improvement}
\label{sec:approx-and-local}

\subsubsection{Assumptions.}
We work with a rooted spanning tree $T$, assuming standard LCA/ancestor metadata is precomputed.
Let $s$ denote the number of sampled trees used in coverage arguments (Theorem~\ref{thm:k-resp-coverage}).
The exact $k$-respecting oracle of Theorem~\ref{thm:exact-eval} returns $\Theta_k(T)$ and a witnessing $k$-respecting cut attaining it.
When we use approximation below, \emph{all} edges are sampled independently (or via Poissonization) with normalization/stratification as stated.

\medskip
\noindent\textbf{Local improvement (certified descent first).}
We consider 1-edge swaps $T' = T-e+f$ with $f\in E\setminus E(T)$ and $e\in P_T(f)$, and define $\mathrm{Score}(T):=\Theta_k(T)$ computed by the exact oracle (Theorem~\ref{thm:exact-eval}).
We accept a swap only if it admits a \emph{certified} improvement and we break ties by a fixed total order on pairs $(e,f)$.

\begin{lemma}[Monotone local improvement]
\label{lem:local}
If a swap $T' = T-e+f$ satisfies $\mathrm{Score}(T') < \mathrm{Score}(T)$, then accepting it yields a strictly decreasing sequence $\mathrm{Score}(T^{(0)})>\mathrm{Score}(T^{(1)})>\cdots$.
Consequently, any sequence of such certified swaps terminates at a 1-swap local optimum.
When approximate screening is used (see below), we validate any tentative improvement with the \emph{exact} oracle (Theorem~\ref{thm:exact-eval}) before committing; monotonicity therefore holds unchanged.
\end{lemma}

\begin{proof}
With exact evaluation, $\Phi(T):=\Theta_k(T)$ is a potential that strictly decreases on accepted swaps.
Since the set of spanning trees is finite, termination follows.
If approximate screening proposes a swap, we compute $\Theta_k(T')$ exactly before acceptance; hence only strictly improving swaps are ever committed.
\end{proof}

\subsubsection{A certified-descent screening rule}
When screening with approximation, the oracle returns an interval
\[
\underline\Theta_k(T)\ :=\ \hat\Theta_k(T)-C_k\varepsilon,\qquad
\overline\Theta_k(T)\ :=\ \hat\Theta_k(T)+C_k\varepsilon,
\]
and accepts a swap $T\to T'$ only if $\overline\Theta_k(T') < \underline\Theta_k(T)$.
Any accepted swap is then re-validated with the \emph{exact} oracle (Theorem~\ref{thm:exact-eval}); the published score and witness cut are therefore exact.

\medskip
\noindent\textbf{Approximate evaluation (pairwise table).}
We approximate the pairwise statistics $\sigma_T(u,v):=w(\delta(D_T(u)\oplus D_T(v)))$ for all $u,v\in V$.

\begin{lemma}[Approximate pairwise table]\label{lem:approx-sigma}
Fix $\varepsilon,\delta\in(0,1)$. Independently sample each edge $e\in E$ with probability
$p=\min\!\left\{1,\;c\,\varepsilon^{-2}\frac{\log(n/\delta)}{|E|}\right\}$,
and assign weight $\tilde w(e)=w(e)/p$ if sampled, else $0$.
Using an $O(\log n)$-time root-to-node update primitive on $T$ (e.g.\ HLD+LCA),
one can compute estimates $\hat\sigma(u,v)$ for all pairs $(u,v)$ in total time
\[
\tilde O\!\Big(|E|\,\varepsilon^{-2}\log(1/\delta)\;+\;n^2\Big)
\]
such that, with probability at least $1-\delta$ simultaneously for all $(u,v)$,
\[
(1-\varepsilon)\,\sigma(u,v)\ \le\ \hat\sigma(u,v)\ \le\ (1+\varepsilon)\,\sigma(u,v).
\]
\end{lemma}

\begin{proof}
Sample each edge $e=\{a,b\}\in E$ independently with probability
$p=\Theta(\varepsilon^{-2}\log(n/\delta)/|E|)$ and assign rescaled weight
$\tilde w(e)=w(e)/p$ if sampled, else $\tilde w(e)=0$.
For a sampled edge $(a,b)$, traverse $a\leadsto \mathrm{lca}(a,b)$ and $b\leadsto \mathrm{lca}(a,b)$ in $T$ and perform the standard path-difference updates to the descendant-indicator counters that realize $\sigma_T(u,v)$; summing over sampled edges yields $\hat\sigma$.
Unbiasedness is immediate.
Concentration follows from bounded-difference or Chernoff bounds (with dyadic weight-stratification if needed); a union bound over $O(n^2)$ pairs (boosted by a constant-round median-of-means) gives the simultaneous guarantee.
The total work is $\tilde O(|E|\varepsilon^{-2}\log(1/\delta))$ for updates plus $O(n^2)$ to materialize/store all pair estimates.
\end{proof}

\medskip
\noindent\textbf{From pairwise approximation to $\Theta_k$.}
For a $k$-respecting descendant cut $A=\bigoplus_{i=1}^k D_T(x_i)$, Lemma~3.2 expresses $w(\delta(A))$ as a signed linear combination of singleton/pairwise primitives with nonnegative coefficients after laminarity reductions.
Writing the inclusion--exclusion over nonempty $I\subseteq[k]$ yields absolute-coefficient mass
\begin{equation}\label{eq:coef-mass}
C_k \;=\;\sum_{\ell=1}^{k} \binom{k}{\ell}\,2^{\ell-1}
\;=\;\frac{3^{k}-1}{2}.
\end{equation}

\begin{lemma}[Approximate evaluation of $\Theta_k$]
\label{lem:approx-theta}
Using the estimates $\hat\sigma$ from Lemma~\ref{lem:approx-sigma}, the oracle’s value $\hat w(\delta(A))$ for any $k$-tuple satisfies
\begin{equation}\label{eq:theta-relerr}
(1 - C_k\varepsilon)\, w(\delta(A)) \ \le\  \hat w(\delta(A)) \ \le\ (1 + C_k\varepsilon)\, w(\delta(A)).
\end{equation}
Consequently,
\begin{equation}\label{eq:Theta-approx}
(1 - C_k\varepsilon)\,\Theta_k(T) \ \le\ \hat\Theta_k(T)\ \le\ (1 + C_k\varepsilon)\,\Theta_k(T).
\end{equation}
For readability, one may write $C_k\varepsilon=O(3^{k}\varepsilon)$; the looser $(1\pm O(2^{k}\varepsilon))$ bound also holds.
\end{lemma}

\begin{proof}
Each primitive is approximated within $(1\pm\varepsilon)$.
By linearity and nonnegativity of the reduction coefficients, the total relative error is at most $C_k\varepsilon$, giving \eqref{eq:theta-relerr}.
Maximizing over $k$-respecting cuts preserves the envelope, giving \eqref{eq:Theta-approx}.
\end{proof}

\medskip
\noindent\textbf{End-to-end guarantee.}
We now combine sampling, certified local improvement, and approximate screening.

\endtoend*

\begin{proof}
Sample $s$ trees from the packing distribution of Theorem~\ref{thm:k-resp-coverage} with $k=c_1\alpha\log n$ and $s=c_2(\alpha\log n+\log(1/\eta))$.
For each $T_i$, build $\hat\sigma$ via Lemma~\ref{lem:approx-sigma} and evaluate $\hat\Theta_k(T_i)$ via Lemma~\ref{lem:approx-theta}.
Return the explicit family $\{(T_i,\Theta_k(T_i))\}_{i=1}^s$ after re-evaluating each $\Theta_k(T_i)$ \emph{exactly} using Theorem~\ref{thm:exact-eval}.
By Theorem~\ref{thm:ensemble-cert} and \eqref{eq:Theta-approx}, with probability at least $1-\eta$ this family forms a verifiable \emph{ensemble certificate} covering all cuts of weight $\le \alpha\lambda$, with per-cut ratio $O((\alpha\log n)/\lambda)$ for some $i$.
The running time is $\tilde O(|E|+n^2+s\cdot n^k)$; for $k=\Theta(\alpha\log n)$ this is quasi-polynomial, while for fixed $k$ it is polynomial.
\end{proof}

\begin{algorithm}[h]
\caption{Thin-Search$(G,k,B)$}
\label{alg:thinsearch}
\KwIn{Graph $G=(V,E,w)$, parameter $k$, iteration budget $B$}
\KwOut{Spanning tree $T^\star$ with $k$-certificate value $\Theta_k(T^\star)$ (for use within the ensemble)}
$T \gets$ arbitrary spanning tree of $G$\;
Precompute LCA/ancestor data on $T$\;
\textsc{PairwiseTable} $\gets$ build all $2$-respecting cut sizes for $T$\;
\textsc{Oracle} $\gets$ construct $k$-respecting oracle(\textsc{PairwiseTable})\;
$C \gets$ seed candidate $k$-tuples\;
$(T^\star,\mathrm{best}) \gets (T,\ \mathrm{Score}(T,C,\textsc{Oracle}))$\;
\For{$it \gets 1$ \KwTo $B$}{
  $f \gets$ random non-tree edge\;
  \ForEach{$e \in P_T(f)$}{
    $T' \gets T-e+f$\;
    update \textsc{PairwiseTable} and \textsc{Oracle} incrementally\;
    $C' \gets$ update candidate tuples\;
    compute $(\underline\Theta_k(T),\overline\Theta_k(T))$ and $(\underline\Theta_k(T'),\overline\Theta_k(T'))$ using $\hat\Theta_k\pm C_k\varepsilon$\;
    \If{$\overline\Theta_k(T') < \underline\Theta_k(T)$}{
       $s_{T'} \gets$ \emph{exact} $\Theta_k(T')$ via Theorem~\ref{thm:exact-eval}\;
       \If{$s_{T'} < \mathrm{best}$}{
           $(T, T^\star, \mathrm{best}) \gets (T', T', s_{T'})$\;
       }
    }
  }
}
Recompute exact witness cut for $T^\star$ (certificate)\;
\Return{$(T^\star,\mathrm{best},\text{witness cut})$}\;
\end{algorithm}

\medskip
\noindent\textbf{Remark on per-swap update cost.}
By Lemma~\ref{lem:swap}, only vertices whose ancestor relation to the fundamental cycle $C=P_T(f)$ changes can affect pairwise values; thus at most $O(|C|\cdot n)$ pairs require updates per swap, keeping local search near-linear in the sum of accepted-cycle lengths.

\section{Planar and Bounded-Genus Graphs \label{sec:special}}
In the general case, our thinness guarantees rely on Karger’s cut-counting
theorem, which gives at most $n^{O(\alpha)}$ cuts of weight $\le \alpha\lambda$.
In planar and bounded-genus graphs, however, classical duality arguments
allow sharper bounds. We recall these known results here, as they lead to
slightly improved sampling parameters, though our main \emph{certificate}
guarantees do not depend on them.

\subsection{Counting near-minimum cuts.}
In planar graphs, every bond corresponds to a simple cycle in the dual $G^\ast$.
It is a standard fact (see, e.g., \cite{Kow03} and related work)
that the number of cycles of length at most $L$ in a planar graph is $O(nL)$.
Combining these observations yields the following folklore bound.

\begin{theorem}[Planar near-min cuts, folklore]
\label{thm:planar-count}
Let $G$ be a simple 2-edge-connected planar graph with edge-connectivity $\lambda$.
For any $\alpha \ge 1$, the number of cuts of value at most $\alpha\lambda$ is $O(n\alpha)$.
Consequently, if we sample
$s = c_2(\alpha\log n + \log(1/\eta))$ trees from a fractional tree packing
distribution and set $k=c_1\alpha\log n$, then with probability at least $1-\eta$,
every such cut is $k$-respecting in at least one sampled tree.
\end{theorem}

\begin{proof}
Each bond in $G$ corresponds to a simple cycle in the dual.
Since planar graphs have only $O(nL)$ cycles of length at most $L$,
the number of bonds of size $\le \alpha\lambda$ is $O(n\alpha)$.
The coverage argument then follows from Theorem~\ref{thm:thm3},
using this smaller union bound.
\end{proof}

A similar argument applies to graphs embedded on surfaces of genus $\gamma$,
where the number of short cycles grows by an additive $O(n\gamma)$ term.

\begin{theorem}[Bounded genus near-min cuts, folklore]
\label{thm:genus-count}
Let $G$ be embedded on an orientable surface of genus $\gamma$.
Then the number of cuts of value at most $\alpha\lambda$ is $O(n(\alpha+\gamma))$.
Consequently, the same coverage guarantee holds with
$s = c_2((\alpha+\gamma)\log n + \log(1/\eta))$ and $k=c_1\alpha\log n$.
\end{theorem}

\medskip
\noindent\textbf{Remark (tighter coverage trade-offs).}
Using Lemma~\ref{lem:packing}, for a fixed cut $A$ with $w(\delta(A))\le\alpha\lambda$ we have
$\Pr[\,|T\cap\delta(A)|>k\,]\le 2\alpha/k$. In planar or bounded-genus graphs, the
\emph{smaller} count of near-min cuts means one can also choose
\[
k=\Theta(\alpha) \quad\text{and}\quad s=\Theta\big(\log n + \log(1/\eta)\big)
\]
(and a mild $\log(\alpha+\gamma)$ factor if desired), trading a slightly larger $s$ for a much smaller $k$.
Either parameterization yields the same ensemble per-cut guarantees below; we keep the
$k=\Theta(\alpha\log n)$ form for consistency with the general-case statements.

\subsection{Dual girth and certified thinness}
The above counting bounds improve sampling efficiency.
A different and complementary phenomenon is that
the \emph{dual girth} directly bounds the certificate parameter.

\begin{theorem}[Dual girth bound]
\label{thm:dual-girth}
Let $G$ be a simple 2-edge-connected planar graph with dual $G^\ast$ of (unweighted) girth $g^\ast$.
For any spanning tree $T$ and any $k\ge 1$,
\[
  \Theta_k(T)\;\le\;\frac{k}{g^\ast}.
\]
In particular, in unweighted planar graphs $g^\ast=\lambda$, and therefore
$\Theta_k(T)\le k/\lambda$. For weighted graphs, if $g^\ast_w$ denotes the
\emph{weighted} dual girth (minimum dual cycle weight), then $g^\ast_w\ge \lambda$
(with equality when all edge weights are $1$), hence $\Theta_k(T)\le k/\lambda$ as well.
\end{theorem}

\begin{proof}
Each fundamental cycle of $T$ in the primal corresponds to a fundamental cut
in the dual. If the dual has girth $g^\ast$, then every nontrivial cycle in $G^\ast$
has length at least $g^\ast$, so every cut in $G$ has cardinality at least $g^\ast$
(unweighted case). If a cut $A$ is $k$-respecting with respect to $T$,
then $|T\cap\delta(A)| \le k$, and thus
$\frac{|T\cap\delta(A)|}{w(\delta(A))} \le k/g^\ast$ (interpreting $w$ as cardinality).
The weighted statement follows by replacing lengths with dual cycle weights.
Maximizing over all $k$-respecting cuts yields the claim.
\end{proof}

\begin{corollary}[Planar certified thinness]
\label{cor:planar-thin}
In unweighted planar graphs, for any spanning tree $T$ and any $k$,
the certificate value satisfies $\Theta_k(T)\le k/\lambda$; for fixed $k$
this is $O(1/\lambda)$. In particular, if $\lambda=\Omega(\log n)$
then $\Theta_k(T)=o(1)$.
\end{corollary}

\begin{theorem}[Ensemble certification in bounded genus]
\label{thm:genus-thin}
Let $G$ be embedded on a surface of genus $\gamma$ with edge-connectivity $\lambda$,
and fix $\alpha\ge 1$ and $\eta\in(0,1)$.
There is a randomized algorithm that samples
\[
s \;=\; \tilde O\big((\alpha+\gamma)\,+\,\log(1/\eta)\big)
\]
trees from a fractional packing distribution and computes exact $k$-certificate
values for each tree with
\[
k \;=\; \Theta(\alpha\log n) \quad\text{(or, using the trade-off above, }k=\Theta(\alpha)\text{)}.
\]
With probability at least $1-\eta$, for \emph{every} cut $A$ of value $\le \alpha\lambda$
there exists a sampled tree $T_i$ such that
\[
\frac{|T_i\cap\delta(A)|}{w(\delta(A))}\ \le\ \Theta_k(T_i)\ \le\ \frac{k}{\lambda}
\;=\; O\!\Big(\tfrac{\alpha\log n}{\lambda}\Big)\quad\text{\rm(or }O(\alpha/\lambda)\text{ under }k=\Theta(\alpha)\text{)}.
\]
Thus, the collection $\{(T_i,\Theta_k(T_i))\}$ forms an explicit, verifiable
\emph{ensemble certificate} covering all $\alpha$-near-minimum cuts.
\end{theorem}

\medskip
In planar and bounded-genus graphs, structural properties
lead to slightly better sampling bounds and, more importantly,
to direct certificate guarantees via dual girth.
Thus our framework not only recovers known cut-counting theorems,
but also strengthens them by producing explicit, verifiable
$k$-respecting certificates.

\section{Conclusion and Future Directions}
\label{sec:summary}
Previous progress on thin trees has largely fallen into two categories:
\emph{existential} results, which establish that thin trees exist but without providing an
efficient construction, and \emph{spectral} results, which prove stronger asymptotic bounds
under spectral relaxations but again without yielding certifiable objects. For example,
Asadpour et al.~\cite{asadpour2017} showed that entropy-rounding a fractional Held--Karp
solution yields an $O(\log n/\log\log n)$-thin tree, but only relative to a fractional LP and
without a certificate that can be verified after the fact. Anari and Oveis Gharan~\cite{anariog2015focs}
proved the existence of spectrally thin trees with thinness $\mathrm{polyloglog}(n)/k$, but
again without a constructive or certifiable guarantee. In planar and bounded-genus graphs,
Oveis Gharan and Saberi~\cite{OS11} obtained $O(\sqrt{\gamma}\log \gamma/k)$ bounds using
duality with cycles, but only at the existential level. More recently, \cite{thin-tree-klien} gave an algorithm to find thin trees restricted to cut constraints that satisfy laminar relations.

\paragraph{What we add.}
We introduce the notion of \emph{certifiable thinness} through the optimization target
$\Theta_k(T)$ and give the first procedure that computes it \emph{exactly} in polynomial time
for fixed $k$, outputting a witnessing cut. This compresses the exponential family of
$k$-respecting cuts to $O(n^2)$ pairwise primitives and turns thinness verification into a
tree-local computation. Combined with fractional tree packings and cut counting, our method
produces an \emph{explicit, verifiable ensemble} of spanning trees with the guarantee that for
\emph{every} $\alpha$-near-minimum cut $A$ there exists some sampled tree $T_i$ certifying
\[
\frac{|T_i\cap\delta(A)|}{w(\delta(A))}\ \le\ \Theta_k(T_i)\ \le\ O\!\Big(\tfrac{\alpha\log n}{\lambda}\Big),
\]
and analogously $O(1/\lambda)$ in planar graphs via dual girth, always with certificates.
The deliverable is thus a compact family $\{(T_i,\Theta_k(T_i))\}$ that \emph{covers all light cuts}
with verifiable bounds, rather than an unproven single-tree global claim.

\paragraph{How this advances the thin-tree program.}
Our results suggest a two-step constructive route toward the thin-tree conjecture:
\begin{enumerate}
  \item \textbf{Coverage at small $k$.} Prove that, in $k$-edge-connected graphs (so $\lambda=k$ in the unweighted case), \emph{every} $\alpha$-near-min cut is $k_0$-respecting in at least one tree from a standard packing with \emph{constant} $k_0=O(1)$ (today we show $k_0=\Theta(\alpha\log n)$). This would immediately upgrade our per-cut certificates to $O(1/\lambda)=O(1/k)$.
  \item \textbf{From ensembles to one tree.} Develop a \emph{certified stitching/patching} procedure (via fundamental-cycle swaps guided by $\Theta_{k_0}$) that merges an ensemble which covers all near-min cuts into a single tree whose $\Theta_{k_0}$ controls all relevant cuts. Our exact oracle and local-improvement primitive provide precisely the feedback needed to attempt such stitching with correctness guarantees.
\end{enumerate}
Either advance would constitute concrete progress toward a constructive resolution of the thin-tree conjecture; achieving both would essentially settle it up to constants.

\paragraph{Concrete next steps.}
We see several promising approaches:
\begin{itemize}
  \item \textbf{Sharper coverage via dependent rounding.} Replace independent tree sampling from a packing by negatively correlated (swap-)rounding to reduce the expected number of crossings per light cut, aiming for \emph{constant} $k$ while keeping concentration.
  \item \textbf{Certified stitching.} Use $\Theta_k$ as a potential to guide fundamental-cycle swaps that monotonically preserve existing per-cut certificates while expanding the set of cuts captured by a \emph{single} tree. The local-update Lemma enables such patching to be implemented and certified.
  \item \textbf{Multi-scale certificates.} Combine $\Theta_{k}$ at geometrically increasing $k$ to control successively heavier cuts; this could yield a composite potential more tightly coupled to global thinness than any single $k$.
  \item \textbf{Exploiting structure.} In planar/bounded-genus graphs, dual girth already implies $\Theta_k(T)\le k/\lambda$ for \emph{every} tree; extending analogous lower bounds (e.g., via expanders/minor decompositions or near-laminarity of near-min cuts) to broader graph classes would directly strengthen certificates.
\end{itemize}

\section{Acknowledgements}
The author was partially supported by ``Stiftelsen för teknisk vetenskaplig forskning och utbildning'' for PhD students at KTH.
\bibliography{Bibliography}
\appendix
\section{Pairwise sufficiency for $\sigma$ on $T$-descendant cuts (Proof of Lemma~\ref{lem:pairwise-sufficiency})}

\label{app:pairwise}

Let $T$ be a rooted spanning tree with root $r_T$. For vertex sets
$A_1,\ldots,A_i\subseteq V$ define
\[
\sigma(A_1,\ldots,A_i)\ \triangleq\ \bigl|\delta(A_1)\cap\cdots\cap\delta(A_i)\bigr|.
\]
In our setting $A_j=D_T(x_j)$ for $x_j\in V\setminus\{r_T\}$.
Recall the descendant family $\{D_T(x):x\neq r_T\}$ is \emph{laminar}:
for any $x,y$, exactly one of $D_T(x)\subseteq D_T(y)$, $D_T(y)\subseteq D_T(x)$,
or $D_T(x)\cap D_T(y)=\emptyset$ holds. We write $x\perp_T y$ iff
$D_T(x)\cap D_T(y)=\emptyset$, and $\level[T]{x}$ for the depth of $x$ in $T$.

\begin{lemma}[Pairwise sufficiency]\label{lemA:pairwise}
Fix $S=\{x_1,\ldots,x_k\}\subseteq V\setminus\{r_T\}$ with $k\ge2$.
Then either
\[
\sigma\bigl(D_T(x_1),\ldots,D_T(x_k)\bigr)=0,
\]
or there exist $p,q\in\{1,\ldots,k\}$ such that
\[
\sigma\bigl(D_T(x_1),\ldots,D_T(x_k)\bigr)
=\sigma\bigl(D_T(x_p),D_T(x_q)\bigr).
\]
\end{lemma}

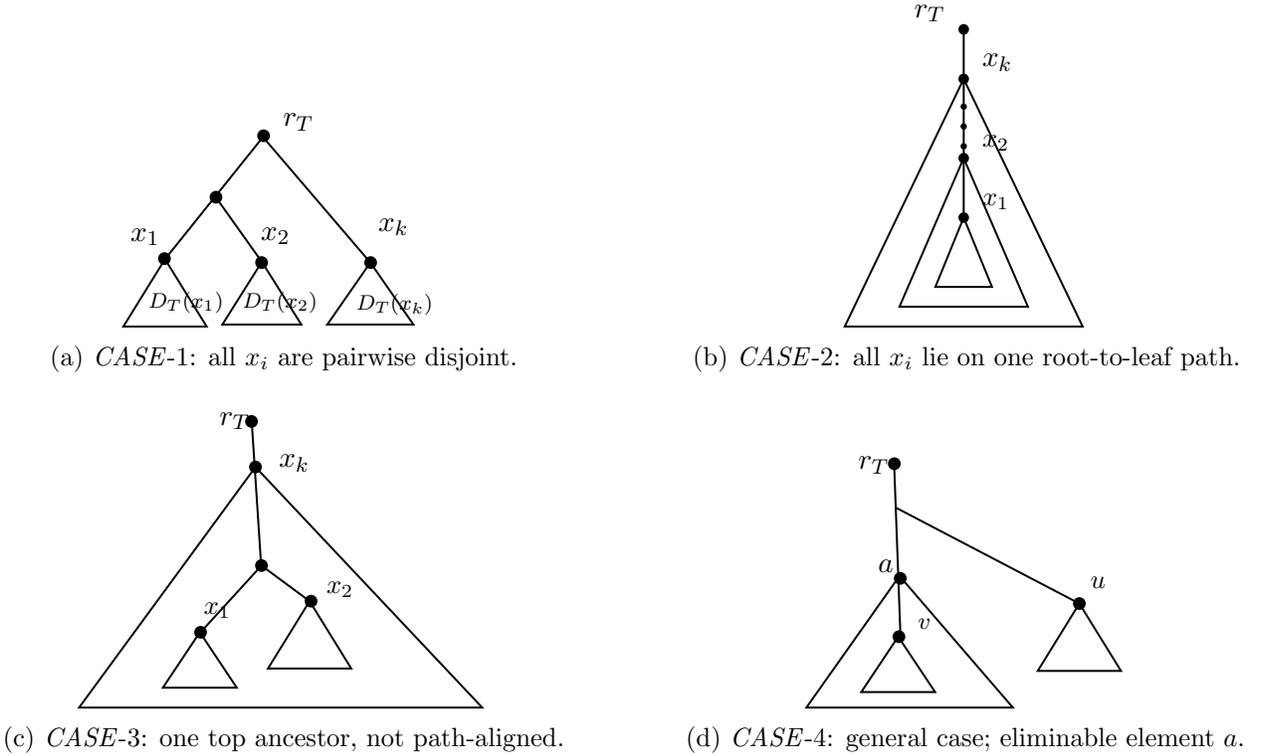
\begin{figure}[h]
  \centering
  \begin{subfigure}{0.45\linewidth}
    \centering
    \tikzset{every picture/.style={line width=0.75pt}} 

\begin{tikzpicture}[x=0.75pt,y=0.75pt,yscale=-1,xscale=1]

\draw (138.46,20.71) -- (87.75,84);
\draw [fill=black] (134.75,21.71) .. controls (134.75,20.22) and (135.97,19) .. (137.46,19)
.. controls (138.96,19) and (140.18,20.22) .. (140.18,21.71)
.. controls (140.18,23.21) and (138.96,24.43) .. (137.46,24.43)
.. controls (135.97,24.43) and (134.75,23.21) .. (134.75,21.71) -- cycle;
\draw (113.11,52.36) -- (136.59,85.36);
\draw (87.75,84) -- (108.75,118) -- (66.75,118) -- cycle;
\draw (136.59,85.36) -- (156.5,117) -- (116.68,117) -- cycle;
\draw (135.75,20.71) -- (191.32,85.71);
\draw (191.32,85.71) -- (213.05,117) -- (169.59,117) -- cycle;
\draw [fill=black] (110.75,52.71) .. controls (110.75,51.22) and (111.97,50) .. (113.46,50)
.. controls (114.96,50) and (116.18,51.22) .. (116.18,52.71)
.. controls (116.18,54.21) and (114.96,55.43) .. (113.46,55.43)
.. controls (111.97,55.43) and (110.75,54.21) .. (110.75,52.71) -- cycle;
\draw [fill=black] (84.75,83.71) .. controls (84.75,82.22) and (85.97,81) .. (87.46,81)
.. controls (88.96,81) and (90.18,82.22) .. (90.18,83.71)
.. controls (90.18,85.21) and (88.96,86.43) .. (87.46,86.43)
.. controls (85.97,86.43) and (84.75,85.21) .. (84.75,83.71) -- cycle;
\draw [fill=black] (133.75,85.75) .. controls (133.73,84.25) and (134.93,83.02) .. (136.43,83)
.. controls (137.93,82.98) and (139.16,84.18) .. (139.18,85.68)
.. controls (139.2,87.18) and (138,88.41) .. (136.5,88.43)
.. controls (135,88.45) and (133.77,87.25) .. (133.75,85.75) -- cycle;
\draw [fill=black] (188.75,85.71) .. controls (188.75,84.22) and (189.97,83) .. (191.46,83)
.. controls (192.96,83) and (194.18,84.22) .. (194.18,85.71)
.. controls (194.18,87.21) and (192.96,88.43) .. (191.46,88.43)
.. controls (189.97,88.43) and (188.75,87.21) .. (188.75,85.71) -- cycle;

\draw (142,5) node [anchor=north west] {$r_{T}$};
\draw (65,63) node [anchor=north west] {$x_{1}$};
\draw (130.75,62.71) node [anchor=north west] {$x_{2}$};
\draw (190,58) node [anchor=north west] {$x_{k}$};

\draw (74,95) node [anchor=north west,font=\scriptsize] {$D_{T}(x_{1})$};
\draw (121.5,95) node [anchor=north west,font=\scriptsize] {$D_{T}(x_{2})$};
\draw (179,96) node [anchor=north west,font=\scriptsize] {$D_{T}(x_{k})$};

\end{tikzpicture}
    \subcaption{\CASE{1}: all $x_i$ are pairwise disjoint.}
    \label{fig:case1}
  \end{subfigure}
  \hfill
  \begin{subfigure}{0.45\linewidth}
    \centering

\tikzset{every picture/.style={line width=0.75pt}} 

\def\xspine{175}    
\def\yroot{25}      

\def\yk{50}         
\def\yii{90}        
\def\yone{120}      

\def\Bk{175}
\def\Bii{165}
\def\Bone{155}

\def\wk{60}
\pgfmathsetmacro{\m}{\wk/(\Bk-\yk)}
\pgfmathsetmacro{\wii}{0.9*\m*(\Bii-\yii)}
\pgfmathsetmacro{\wone}{0.85*\m*(\Bone-\yone)}

\begin{tikzpicture}[x=0.75pt,y=0.75pt,yscale=-1,xscale=1]

\draw (\xspine,\yroot) -- (\xspine,\yone);

\fill (\xspine,\yroot) circle (2pt);   
\fill (\xspine,\yk)   circle (2pt);    
\fill (\xspine,\yii)  circle (2pt);    
\fill (\xspine,\yone) circle (2pt);    

\foreach \y in {64, 74, 84} {\fill (\xspine,\y) circle (1.2pt);}

\draw (\xspine,\yk)   -- (\xspine+\wk,\Bk)     -- (\xspine-\wk,\Bk)     -- cycle; 
\draw (\xspine,\yii)  -- (\xspine+\wii,\Bii)   -- (\xspine-\wii,\Bii)   -- cycle; 
\draw (\xspine,\yone) -- (\xspine+\wone,\Bone) -- (\xspine-\wone,\Bone) -- cycle; 

\node[above right,xshift=3pt,yshift=-1pt]             at (\xspine,\yk)   {$x_{k}$};
\node[above right,xshift=3pt,yshift=-1pt,font=\small] at (\xspine,\yii)  {$x_{2}$};
\node[above right,xshift=3pt,yshift=-1pt,font=\small] at (\xspine,\yone) {$x_{1}$};
\node[above left,xshift=-3pt,yshift=-1pt]             at (\xspine,\yroot) {$r_{T}$};

\end{tikzpicture}
    \subcaption{\CASE{2}: all $x_i$ lie on one root-to-leaf path.}
    \label{fig:case2}
  \end{subfigure}

  \vspace{0.8em}

  \begin{subfigure}{0.45\linewidth}
    \centering
    \tikzset{every picture/.style={line width=0.75pt}} 

\begin{tikzpicture}[x=0.75pt,y=0.75pt,yscale=-1,xscale=1]

\draw (173.43,12.57) -- (178.49,89.74);


\draw (147.38,120.71) -- (166,148.67) -- (128.75,148.67) -- cycle;
\draw (175.43,37.29) -- (289.8,158.71) -- (86.48,158.71) -- cycle;

\draw [fill=black] (172.71,37.29) .. controls (172.71,35.79) and (173.93,34.57) .. (175.43,34.57)
.. controls (176.93,34.57) and (178.14,35.79) .. (178.14,37.29)
.. controls (178.14,38.78) and (176.93,40) .. (175.43,40)
.. controls (173.93,40) and (172.71,38.78) .. (172.71,37.29) -- cycle ;
\draw [fill=black] (145.04,120.71) .. controls (145.04,119.22) and (146.25,118) .. (147.75,118)
.. controls (149.25,118) and (150.46,119.22) .. (150.46,120.71)
.. controls (150.46,122.21) and (149.25,123.43) .. (147.75,123.43)
.. controls (146.25,123.43) and (145.04,122.21) .. (145.04,120.71) -- cycle ;
\draw [fill=black] (175.75,87.06) .. controls (175.73,85.56) and (176.93,84.33) .. (178.43,84.32)
.. controls (179.93,84.3) and (181.16,85.5) .. (181.18,87)
.. controls (181.19,88.5) and (179.99,89.73) .. (178.49,89.74)
.. controls (176.99,89.76) and (175.77,88.56) .. (175.75,87.06) -- cycle ;
\draw [fill=black] (170.71,14.29) .. controls (170.71,12.79) and (171.93,11.57) .. (173.43,11.57)
.. controls (174.93,11.57) and (176.14,12.79) .. (176.14,14.29)
.. controls (176.14,15.78) and (174.93,17) .. (173.43,17)
.. controls (171.93,17) and (170.71,15.78) .. (170.71,14.29) -- cycle ;
\draw [fill=black] (200.75,105) .. controls (200.75,103.5) and (201.97,102.29) .. (203.46,102.29)
.. controls (204.96,102.29) and (206.18,103.5) .. (206.18,105)
.. controls (206.18,106.5) and (204.96,107.71) .. (203.46,107.71)
.. controls (201.97,107.71) and (200.75,106.5) .. (200.75,105) -- cycle ;

\draw (178.46,87.03) -- (147.75,120.71);
\draw (202.75,105) -- (223.75,139) -- (181.75,139) -- cycle;
\draw (178.46,87.03) -- (202.75,105);

\draw (152,4.4) node [anchor=north west] {$r_{T}$};
\draw (144,100.4) node [anchor=north west,font=\small] {$x_{1}$};
\draw (205.95,90.04) node [anchor=north west,font=\small] {$x_{2}$};
\draw (182,25.4) node [anchor=north west] {$x_{k}$};


\end{tikzpicture}
    \subcaption{\CASE{3}: one top ancestor, not path-aligned.}
    \label{fig:case3}
  \end{subfigure}
  \hfill
  \begin{subfigure}{0.45\linewidth}
    \centering
    \tikzset{every picture/.style={line width=0.75pt}} 

\begin{tikzpicture}[x=0.75pt,y=0.75pt,yscale=-1,xscale=1]

\draw (152.43,18.14) -- (155.43,99.71);

\draw (154.38,102.71) -- (173,130.67) -- (135.75,130.67) -- cycle;
\draw (155.03,72.07) -- (212.21,138.5) -- (108.08,138.5) -- cycle;

\draw [fill=black] (152.71,73.29) .. controls (152.71,71.79) and (153.93,70.57) .. (155.43,70.57)
.. controls (156.93,70.57) and (158.14,71.79) .. (158.14,73.29)
.. controls (158.14,74.78) and (156.93,76) .. (155.43,76)
.. controls (153.93,76) and (152.71,74.78) .. (152.71,73.29) -- cycle;
\draw [fill=black] (152.04,102.71) .. controls (152.04,101.22) and (153.25,100) .. (154.75,100)
.. controls (156.25,100) and (157.46,101.22) .. (157.46,102.71)
.. controls (157.46,104.21) and (156.25,105.43) .. (154.75,105.43)
.. controls (153.25,105.43) and (152.04,104.21) .. (152.04,102.71) -- cycle;
\draw [fill=black] (149.71,15.43) .. controls (149.71,13.93) and (150.93,12.71) .. (152.43,12.71)
.. controls (153.93,12.71) and (155.14,13.93) .. (155.14,15.43)
.. controls (155.14,16.93) and (153.93,18.14) .. (152.43,18.14)
.. controls (150.93,18.14) and (149.71,16.93) .. (149.71,15.43) -- cycle;
\draw [fill=black] (243.04,86) .. controls (243.04,84.5) and (244.25,83.29) .. (245.75,83.29)
.. controls (247.25,83.29) and (248.46,84.5) .. (248.46,86)
.. controls (248.46,87.5) and (247.25,88.71) .. (245.75,88.71)
.. controls (244.25,88.71) and (243.04,87.5) .. (243.04,86) -- cycle;

\draw (153.43,37.71) -- (245.75,86);
\draw (245.75,86) -- (266.75,120) -- (224.75,120) -- cycle;


\draw (129,6.4) node [anchor=north west] {$r_{T}$};
\draw (159,88.4) node [anchor=north west,font=\footnotesize] {$v$};
\draw (245.95,67.04) node [anchor=north west,font=\small] {$u$};
\draw (139,58.4) node [anchor=north west] {$a$};


\end{tikzpicture}
    \subcaption{\CASE{4}: general case; eliminable element $a$.}
    \label{fig:case4}
  \end{subfigure}

  \caption{Ancestor–descendant configurations for the four cases used in the proof of Lemma~\ref{lem:pairwise-sufficiency}.}
  \label{fig:four-cases}
\end{figure}

\begin{proof}
For an edge $e=\{u,v\}$ and a set $A$, $e\in\delta(A)$ iff exactly one of $\{u,v\}$ lies in $A$.
Hence
\begin{equation}\label{eq:iff}
e\in\bigcap_{j=1}^k \delta(A_j)
\iff
\forall j\in[k]\; \bigl|\{u,v\}\cap A_j\bigr|=1.
\end{equation}
When the $A_j$ are laminar, \eqref{eq:iff} forces the following:
(i) if $A\subset B$, then any edge cut by \emph{both} $\delta(A)$ and $\delta(B)$
has one endpoint in $A$ and the other in $V\setminus B$;
(ii) if $A\cap B=\emptyset$, then any edge in $\delta(A)\cap\delta(B)$
has one endpoint in $A$ and the other in $B$. Assume $k\ge3$ unless noted.
\smallskip
\noindent\CASE{1} (all $x_i$ pairwise disjoint).
If $x_i\perp_T x_j$ for all $i\neq j$, any edge in
$\bigcap_{i=1}^k \delta(D_T(x_i))$ would need one endpoint in
\emph{each} distinct $D_T(x_i)$, impossible for $k\ge3$. Thus the
intersection is empty and $\sigma=0$.
\medskip
\noindent\CASE{2} (all $x_i$ lie on one root--to--leaf path).
Order $S$ by depth and let $x_q$ be shallowest and $x_p$ deepest.
Then $D_T(x_p)\subset\cdots\subset D_T(x_q)$ and
$V\setminus D_T(x_p)\supset\cdots\supset V\setminus D_T(x_q)$.
We claim
\[
\bigcap_{i=1}^k \delta(D_T(x_i))=\delta(D_T(x_p))\cap\delta(D_T(x_q)).
\]
The $\subseteq$ direction is immediate. For $\supseteq$, if
$e\in\delta(D_T(x_p))\cap\delta(D_T(x_q))$ then its endpoints are
$u\in D_T(x_p)$ and $v\in V\setminus D_T(x_q)$. For any intermediate
$D_T(x_i)$ we have $u\in D_T(x_i)$ and $v\notin D_T(x_i)$, so
$e\in\delta(D_T(x_i))$. Hence
$\sigma(D_T(x_1),\ldots,D_T(x_k))=\sigma(D_T(x_p),D_T(x_q))$.
\medskip
\noindent\CASE{3} (one top ancestor but not path-aligned).
Suppose there exists $x^\top$ with $S\setminus\{x^\top\}\subset D_T(x^\top)$
but the vertices in $S$ are not all on a single root--to--leaf path.
Let $x_1$ be deepest in $S$ and choose $x_2$ that is not an ancestor of $x_1$.
Then $D_T(x_1),D_T(x_2)\subset D_T(x^\top)$ and $D_T(x_1)\cap D_T(x_2)=\emptyset$.
Any $e\in\delta(D_T(x_1))\cap\delta(D_T(x_2))$ has endpoints split across
$D_T(x_1)$ and $D_T(x_2)$, so both endpoints lie inside $D_T(x^\top)$,
hence $e\notin\delta(D_T(x^\top))$. Therefore
$\delta(D_T(x_1))\cap\delta(D_T(x_2))\cap\delta(D_T(x^\top))=\emptyset$,
and the $k$-wise intersection is empty: $\sigma=0$.
\medskip
\noindent\CASE{4} (general case; eliminable element).
If none of Cases 1--3 applies, then there exist $a,v\in S$ with
$D_T(v)\subset D_T(a)$ and some $u\in S$ with $u\notin D_T(a)$.
Choose such a pair with $a$ as shallow as possible; then $u\perp_T a$,
and consequently $u\perp_T v$ as well.

We claim
\begin{equation}\label{eq:elim}
\delta(D_T(a))\cap\delta(D_T(v))\cap\delta(D_T(u))
\;=\;
\delta(D_T(v))\cap\delta(D_T(u)).
\end{equation}
Indeed, for $\supseteq$ take $e\in\delta(D_T(v))\cap\delta(D_T(u))$.
Since $D_T(v)\subset D_T(a)$ and $u\perp_T a$, $e$ has endpoints
$x\in D_T(v)\subset D_T(a)$ and $y\in D_T(u)\subset V\setminus D_T(a)$,
so $e\in\delta(D_T(a))$. The reverse inclusion is trivial.
Thus removing $D_T(a)$ from the tuple does not change the intersection:
\[
\bigcap_{x\in S}\delta(D_T(x))
=
\bigcap_{x\in S\setminus\{a\}}\delta(D_T(x)).
\]
Iterating this elimination strictly reduces $k$ and must terminate
in one of the previous cases, which yield either $0$ (Cases 1 or 3)
or a pairwise intersection (Case 2).
\end{proof}
\smallskip
\textbf{Consequence (for weighted graphs).}
If edges carry nonnegative weights $w:E\to\mathbb{R}_{\ge0}$, the same
reasoning applies verbatim with $\sigma$ interpreted as the total weight
of the intersection, by linearity over edges.
\section{Computing all $\sigma(u,v)$ in $O(m\log n+n^2)$ time}
\label{app:sigma}

\noindent\textbf{Goal.}
For a tree $T$ rooted at $r_T$ and a weighted graph $(G,w)$ on the same vertex set,
compute
\[
\sigma(u,v)\;=\;w\!\bigl(\delta(D_T(u)\oplus D_T(v))\bigr)
\quad\text{for all } u,v\in V.
\]
We prove a preprocessing bound of $O(m\log n+n^2)$ time (or $O(m+n^2)$ offline) and $O(n^2)$ space.

\medskip
\noindent\textbf{Tree primitives.}
Compute Euler tour entry/exit times $\mathrm{tin}(\cdot),\mathrm{tout}(\cdot)$, depths, and an LCA structure.
We write $x\preceq y$ iff $x$ is an ancestor of $y$ in $T$.
We also fix either (i) a heavy--light decomposition (HLD) to support $O(\log n)$
root-to-node path additions, or (ii) an offline DFS accumulation that achieves total $O(m+n)$ for all path updates.

\medskip
\noindent\textbf{Key indicator identity.}
For an edge $e=\{a,b\}\in E(G)$ and a tree vertex $u$,
\[
e\in\delta(D_T(u))
\quad\Longleftrightarrow\quad
(u\preceq a)\ \oplus\ (u\preceq b),
\]
i.e., $u$ is an ancestor of exactly one endpoint of $e$.
Equivalently, if $U_e:=\mathrm{Anc}(a)\triangle \mathrm{Anc}(b)$,
then $\mathbf{1}[e\in\delta(D_T(u))]=\mathbf{1}[u\in U_e]$.

\medskip
\noindent\textbf{From $\sigma$ to unary and pairwise counts.}
Define
\[
\tau(u)\;:=\;\sum_{e\in E} w(e)\,\mathbf{1}[u\in U_e]
\qquad\text{and}\qquad
\pi(u,v)\;:=\;\sum_{e\in E} w(e)\,\mathbf{1}[u\in U_e]\mathbf{1}[v\in U_e].
\]
Then, for all $u,v$,
\begin{equation}
\label{eq:sigma-decomp}
\sigma(u,v)\;=\;\tau(u)+\tau(v)-2\,\pi(u,v).
\end{equation}
Thus it suffices to compute all $\tau(\cdot)$ and all $\pi(\cdot,\cdot)$.

\medskip
\noindent\textbf{Step 1: computing all $\boldsymbol{\tau(u)}$ in $O(m\log n)$ (or $O(m)$) time.}
For each $e=\{a,b\}$ with weight $w$, let $\ell=\mathrm{lca}(a,b)$.
Perform three root-to-node \emph{path-additions} of value $w$: add $+w$ on
$r_T\rightsquigarrow a$ and on $r_T\rightsquigarrow b$, and add $-2w$ on $r_T\rightsquigarrow \ell$.
With HLD, each edge contributes $O(\log n)$; total $O(m\log n)$.
In the offline variant, mark endpoints $a,b$ with $+w$ and $\ell$ with $-2w$, and propagate sums to ancestors in a single DFS, achieving $O(m+n)$ total.
A standard inclusion–exclusion argument shows that the final value stored at node $u$ equals
\[
\tau(u)=\sum_{e} w(e)\,\bigl(\mathbf{1}[u\preceq a]+\mathbf{1}[u\preceq b]-2\mathbf{1}[u\preceq \ell]\bigr)
=\sum_{e} w(e)\,\mathbf{1}[u\in U_e]
= w\!\bigl(\delta(D_T(u))\bigr).
\]

\medskip
\paragraph{Step 2: From $\pi$ to a 2D ancestor–ancestor sum.}
Write $U_e := \Anc(a)\triangle \Anc(b)$ for $e=\{a,b\}$ and $\ell=\operatorname{lca}(a,b)$.
Then
\[
1_{U_e} \;=\; 1_{\Anc(a)} + 1_{\Anc(b)} - 2\,1_{\Anc(\ell)}.
\]
Hence, expanding $1_{U_e}(u)\,1_{U_e}(v)$,
\[
\pi(u,v)\;=\;\sum_{e\in E} w(e)\,1_{U_e}(u)\,1_{U_e}(v)
\;=\;\sum_{x,y\in V}\beta[x,y]\;1[u\preceq x]\;1[v\preceq y],
\]
where the $n\times n$ coefficient table $\beta$ is initialized to $0$ and updated for each edge $e=\{a,b\}$ with LCA $\ell$ by:
\[
\begin{aligned}
&\beta[a,a]{+}{=}\,w(e),\quad \beta[b,b]{+}{=}\,w(e),\quad \beta[a,b]{+}{=}\,w(e),\quad \beta[b,a]{+}{=}\,w(e),\\
&\beta[a,\ell]{-}{=}\,2w(e),\ \beta[\ell,a]{-}{=}\,2w(e),\ \beta[b,\ell]{-}{=}\,2w(e),\ \beta[\ell,b]{-}{=}\,2w(e),\\
&\beta[\ell,\ell]{+}{=}\,4w(e).
\end{aligned}
\]

\paragraph{Step 3: Computing all $\pi(u,v)$ in $O(n^2)$ time.}
Define
\[
F(u,v)\ :=\ \sum_{x\in\Sub(u)}\sum_{y\in\Sub(v)} \beta[x,y].
\]
Then $\pi(u,v)=F(u,v)$ for all $u,v$.
Compute two auxiliary tables in $O(n^2)$ time:
\[
\RowSum(u,v):=\sum_{x\in\Sub(u)}\beta[x,v],\qquad
\ColSum(u,v):=\sum_{y\in\Sub(v)}\beta[u,y].
\]
(For each fixed $v$, compute $\RowSum(\cdot,v)$ by a single bottom-up pass on $T$; analogously for $\ColSum$ with $u$ fixed.)
Now evaluate $F(\cdot,\cdot)$ bottom-up on nondecreasing $\max\{\depth(u),\depth(v)\}$ using the partition
\[
\Sub(u)\times\Sub(v)=\{(u,v)\}\;\dot\cup\;\Big(\!\bigcup_{p\in S(u)} \Sub(p)\times\{v\}\!\Big)\;\dot\cup\;\Big(\!\bigcup_{q\in S(v)} \{u\}\times\Sub(q)\!\Big)\;\dot\cup\;\Big(\!\bigcup_{p\in S(u)}\bigcup_{q\in S(v)} \Sub(p)\times\Sub(q)\!\Big),
\]
which yields the recurrence
\[
F(u,v)\;=\;\beta[u,v]\;+\;\sum_{p\in S(u)} \RowSum(p,v)\;+\;\sum_{q\in S(v)} \ColSum(u,q)\;+\;\sum_{p\in S(u)}\sum_{q\in S(v)} F(p,q).
\]
Every pair $(p,q)$ contributes to exactly one parent $(u=\parent(p),v=\parent(q))$, so the total work over all $(u,v)$ is $O(n^2)$.
Finally, set $\sigma(u,v)=\tau(u)+\tau(v)-2F(u,v)$.

\medskip
\noindent\textbf{Step 4: assemble $\boldsymbol{\sigma}$.}
Finally, for every $(u,v)$ set
$\sigma(u,v)=\tau(u)+\tau(v)-2\,G(u,v)$ using \eqref{eq:sigma-decomp}.
This pass is $O(n^2)$ and the table occupies $O(n^2)$ space.

\medskip
\noindent\textbf{Correctness.}
Step~1 yields $\tau(u)=w(\delta(D_T(u)))$ by the indicator identity and linearity over edges.
Step~2 expresses $\pi(u,v)$ as a weighted count over vertices $x$ that are descendants of both $u$ and $v$.

\end{document}